\documentclass[letterpaper, 10pt, conference]{ieeeconf}
\IEEEoverridecommandlockouts
\usepackage{geometry}
\usepackage{graphicx}
\usepackage{epsfig}
\usepackage{amsmath}
\usepackage{amssymb} 
\usepackage[noadjust]{cite}
\usepackage{bm}
\usepackage{subfigure}
\usepackage{acronym}
\usepackage{paralist}
\usepackage{float}
\usepackage{epstopdf}
\usepackage{algorithm}
\usepackage[noend]{algpseudocode}
\usepackage{mathtools}
\usepackage{multicol} 
\usepackage{accents}
\usepackage{breqn}

\makeatletter
\def\BState{\State\hskip-\ALG@thistlm}
\makeatother

\newtheorem{define}{Definition}
\newtheorem{theorem}{Theorem}
\newtheorem{lemma}{Lemma}

\newcommand{\R}{\mathbb R}
\newcommand{\N}{\mathcal{N}}
\newcommand{\D}{\mathcal{D}}
\newcommand{\V}{\mathcal{V}}

\newcommand{\E}{\mathcal{E}}
\newcommand{\Xr}[2]{\mathcal{X}^{#1}_{S_{#2}}}

\newcommand{\Le}{\mathcal{L}}
\newcommand{\A}{\mathcal{A}}
\newcommand{\J}{\mathcal{J}}
\newcommand{\eps}{\epsilon}
\newcommand{\mb}{\overline{m}}
\newcommand{\Mb}{\overline{M}}
\newcommand{\Z}{\mathbb{Z}}

\newcommand{\nrm}[1]{\left \lVert#1\right \rVert}

\DeclarePairedDelimiter{\ceil}{\lceil}{\rceil}
\DeclarePairedDelimiter{\floor}{\lfloor}{\rfloor}
\DeclarePairedDelimiter{\abs}{\lvert}{\rvert}

\newcommand{\rarr}{\rightarrow} 

\makeatletter
\let\oldceil\ceil
\def\ceil{\@ifstar{\oldceil}{\oldceil*}}
\makeatother
\makeatletter
\let\oldfloor\floor
\def\floor{\@ifstar{\oldfloor}{\oldfloor*}}
\makeatother
\makeatletter
\let\oldnorm\norm
\def\norm{\@ifstar{\oldnorm}{\oldnorm*}}
\makeatother
\makeatletter
\let\oldabs\abs
\def\abs{\@ifstar{\oldabs}{\oldabs*}}
\makeatother

\begin{document}

\title{\LARGE \bf Resilient Leader-Follower Consensus to Arbitrary Reference Values}

\author{James Usevitch and Dimitra Panagou
\thanks{James Usevitch is with the Department of Aerospace Engineering, University of Michigan, Ann Arbor; \texttt{usevitch@umich.edu}.
Dimitra Panagou is with the Department of Aerospace Engineering, University of Michigan, Ann Arbor; \texttt{dpanagou@umich.edu}.
The authors would like to acknowledge the support of the Automotive Research Center (ARC) in accordance with Cooperative Agreement W56HZV-14-2-0001 U.S. Army TARDEC in Warren, MI. This work has been funded by the Center for Unmanned Aircraft Systems (C-UAS), a National Science Foundation Industry/University Cooperative Research Center (I/UCRC) under NSF Award No. 1738714 along with significant contributions from C-UAS industry members. }
}

\maketitle
\thispagestyle{empty}
\pagestyle{empty}

\acrodef{wrt}[w.r.t.]{with respect to}
\acrodef{apf}[APF]{Artificial Potential Fields}
\begin{abstract}
The problem of consensus in the presence of misbehaving agents has increasingly attracted attention in the literature. Prior results have established algorithms and graph structures for multi-agent networks which guarantee the consensus of normally behaving agents in the presence of a bounded number of misbehaving agents. The final consensus value is guaranteed to fall within the convex hull of initial agent states. However, the problem of consensus tracking considers consensus to arbitrary reference values which may not lie within such bounds. Conditions for consensus tracking in the presence of misbehaving agents has not been fully studied. This paper presents conditions for a network of agents using the W-MSR algorithm to achieve this objective.
\end{abstract}

\IEEEpeerreviewmaketitle

\section{Introduction}

There has been an increasing amount of recent work on resilient consensus-based algorithms. Several authors have proposed algorithms that employ local filtering mechanisms to guarantee that normal agents are able to come to agreement in the presence of a bounded number of malicious or misbehaving agents. These algorithms include the ARC-P, W-MSR, SW-MSR, and DP-MSR algorithms (\cite{Leblanc2013a,LeBlanc_2013_Res,Salda2017,Dibaji2017resilient}). These algorithms are able to mitigate the effects of malicious and misbehaving agents without the need for normal agents to explicitly identify the sources of misbehavior. They guarantee that the consensus value of normally behaving agents is within the convex hull of the initial agent state values under the assumption of $r$-robustness or $(r,s)$-robustness of the network. On the other hand, consensus-based algorithms without resilience to malicious agents have been employed in the literature to solve the problem of coming to consensus to reference values, where agents seek to come to agreement with a reference state which may or may not lie within the convex hull of initial state values (see \cite{Ren2008book,Moore2007,Mesbahi2010} and references). Solutions to this problem have been studied in the literature under the assumption that agents are behaving nominally (e.g. \cite{Ren2007a,Moore2007}). 

The problem of consensus to reference values in $r$-robust networks has not been as thoroughly treated. The most relevant work is that in \cite{LeBlanc2013Sync,LeBlanc2017,leblanc2014,Mitra2016secure,mitra2018}. In \cite{LeBlanc2013Sync} and \cite{LeBlanc2017}, a resilient protocol and dynamic feedback laws are given to synchronize trajectories of continuous-time agents to the same stable zero-input solution of a time-invariant system in the presence of adversaries. This work only considers continuous-time LTI systems however, with the assumption that the time-varying reference state satisfies the dynamics $\dot{x}_r = Ax_r$ for some $A$. In \cite{leblanc2014}, the problem of resilient distributed estimation is considered where agents employ a discrete-time resilient consensus algorithm similar to the W-MSR algorithm to reduce the estimation error of individual parameters of interest. The authors assume that certain nodes have a precise knowledge of their own parameters. These "reliable nodes" drive the errors of the rest of the normal nodes to the reference value of zero in the presence of misbehaving agents. The main result of this paper is that consensus of the normal agents' error values to zero occurs if for each normal node $i$ there exists an infinite sequence of time periods (with bounded, finite time between each time period) where $i$ interacts with at least $F+1$ reliable nodes. A limitation of these results is the interaction requirement between each normal node and the minimum required number of reliable nodes, which becomes increasingly difficult to satisfy as the network size increases. The problem of consensus with a static reference state also somewhat resembles the resilient broadcast problem (\cite{Zhang2012b,Pagourtzis2017,Pelc2005,Koo2004a}) in which a source node seeks to transmit a message to all other nodes in the presence of malicious or misbehaving nodes. However, the state update algorithms differ from those commonly applied in $r$-robust networks. Finally, in \cite{Mitra2016secure,mitra2018} the authors consider resilient distributed estimation of a system of the form $x[t+1] = Ax[t]$ where each individual agent may only be able to observe part of the system modes. A \emph{Mode Estimation Directed Acyclic Graph} (MEDAG) is constructed for each mode $\lambda_j$ of $A$. Each MEDAG consists of a set of at least $(2F+1)$ source nodes which can observe $\lambda_j$, and successive subsets (called \emph{levels}) of the remaining nodes with sufficient in-neighbors to guarantee information flow from the source nodes. If there exists a MEDAG for each unobservable mode and if agents can explicitly identify in-neighbors from preceding levels, a resilient consensus algorithm can be used to bring agents' estimates of the system state into agreement with the actual state. A graph condition called \emph{strong $r$-robustness} (\cite{Zhang2012robustness}) is a sufficient condition for a graph to contain a MEDAG; specifically, the network must be strongly $(2F+1)$-robust with respect to a set of source nodes. However, the algorithm given in \cite{Mitra2016secure} for agents to identify in-neighbors from preceding levels requires the graph to initially be strongly $(3F+1)$-robust. 


The contributions of this paper differ from prior literature in the following ways: first, we consider consensus tracking with a reference signal which is a piecewise continuous step function. We show that under the W-MSR algorithm, consensus tracking is guaranteed if the graph is strongly $(2F+1)$-robust without the need for agents to identify in-neighbors from preceding levels in a MEDAG. Since this identification is not required, the graph does not need to initially be $(3F+1)$-robust. Second, we incorporate the notion of trusted nodes (or trusted agents) into resilient consensus tracking \cite{Abbas2014}. Trusted agents are agents which have been sufficiently secured such that it can be safely assumed that they will always behave. Having trusted agents as leaders allows the condition of strong $r$-robustness to be slightly relaxed while still guaranteeing consensus to the reference value. Finally, there is very little prior work on proving that any class of graph is strongly $r$-robust. We demonstrate that certain graphs called $k$-circulant graphs are strongly $r$-robust when a proper subset of agents are chosen to behave as leaders.


Our paper is organized as follows: Section 2 contains our notation and problem formulation. In Section 3 we present necessary conditions for consensus tracking in robust networks and the insufficiency of $r$-robustness and $(r,s)$-robustness to guarantee consensus tracking to arbitrary values. Section 4 contains our main results on consensus to a reference state in the presence of adversaries. In Section 5 we demonstrate the strong $r$-robustness of $k$-circulant graphs.  Section 6 contains simulations demonstrating our work, and in section 7 we give a brief conclusion.

\label{intro}

\section{Notation and Problem Formulation}

We denote a graph as $\mathcal{G} = (\mathcal{V},\E)$ and a digraph as $\mathcal{D} = (\mathcal{V},\E)$, with $\V = \{1,...,n\}$ denoting the vertex set, or agent set, of the graph and $\E$ denoting the edge set of the graph. A directed edge ${(i,j) \in \E}: i,j \in \V$ denotes that there exists a connection from agent $i$ to agent $j$. Agent $j$ is able to receive information from agent $i$ if $(i,j)$ is in $\E$ (note that $(i,j) \neq (j,i)$). This implies that agent $i$ is an \emph{in-neighbor} of $j$ and agent $j$ is an \emph{out-neighbor} of $i$. We denote the set of in-neighbors of agent $i$ as $\V_i = \{j: (j,i) \in \E)\}$ and the set of inclusive neighbors  of $i$ as $\J_i = (\V_i \cup \{i\})$ \cite{LeBlanc_2013_Res}. We denote the cardinality of a set $S$ as $|S|$, the set of integers as $\mathbb{Z}$, the set of integers greater than or equal to 0 as $\mathbb{Z}_{\geq 0}$. and the natural numbers as $\mathbb{N}$.

An undirected graph of $n$ agents is called circulant if there exists a set $\{a_1, a_2, \ldots, a_l \in \mathbb{Z}_{\geq 0}: a_1 < a_2 < \ldots < a_l < n\}$ such that $(i, \left[i \pm a_1 \right] \text{mod}\, n) \in \E, \ldots, (i,\left[i \pm a_l \right] \text{mod}\, n) \in \E$ \cite{Boesch1984}. We call such a graph an \emph{undirected circulant graph} and denote it as $C_n(\pm a_1, \pm a_2, \ldots, \pm a_m) = (\V, \E)$. These graphs are constructed over the additive group of integers modulo $n$ (the agents $n+a$ and $a$ are congruent modulo $n$). Similarly, we call a digraph of $n$ agents \emph{circulant} if there exists a set $\{a_1, a_2, \ldots, a_m: 0 < a_1 < a_2 < \ldots < a_m < n\},\ m \in \mathbb{Z}_{\geq 0}$ such that $(i, \left[i+a_1 \right] \text{mod}\, n) \in \mathcal{E}, \ldots, (i,\left[i+a_m \right] \text{mod}\, n) \in \mathcal{E}$. We denote such a graph as $C_n(a_1, a_2, \ldots, a_m) = (\V, \E)$ and call it a \emph{directed circulant graph} or \emph{circulant digraph}.

In this paper, we will consider a specific class of circulant digraphs called \emph{k-circulant digraphs}:

\begin{define}
Let $n \in \mathbb{Z},\ n \geq 2$ and let $k \in \mathbb{Z}: 1 \leq k \leq n-1$. A $k$-circulant digraph  is any circulant digraph of the form $C_n(1,2,3,\ldots,k) = (\V, \E)$.
\end{define}

This type of graph is fully determined by the number of agents $n$ and by the parameter $k$, which determines the in- and out-neighbors of each agent.

Finally, the notions of reachability, \emph{r}-robustness, and $(r,s)$-robustness were defined by the authors of \cite{LeBlanc_2013_Res}. We do not repeat the definitions here due to lack of space. Although the definitions refer specifically to digraphs, they also apply to undirected graphs.\footnote{Undirected graphs can be modeled as digraphs in which $(i,j) \in \E \iff (j,i) \in \E$}

\subsection{Problem Definition}

The prior literature has mainly dealt with \emph{resilient asymptotic consensus} (\cite{LeBlanc_2013_Res,Zhang2015a}), where agents remain within the convex hull of their initial values and come to consensus on a common value. However, there may exist cases where it is desired for agents' states to reach consensus on some reference value that may be outside the convex set of initial states. To address this problem, we assume there exists a potentially time-varying reference state $x_r[t]$ that satisfies $x_r[t+1] = f_r(t,x_r[t])$. We assume that $f_r(t,x_r[t])$ is a piecewise-continuous step function where for each point of discontinuity $t_d$, $f_r(t_d,x_r[t_d]) = \lim_{t \rarr t_d^+} f_r(t,x_r[t])$. The objective of consensus tracking is
\begin{equation}
\lim_{t \rightarrow \infty} \nrm{x_i[t] - x_r[t]} = 0,\ \forall i \in \V
\end{equation}

We consider the case where the reference state is constant after some finite time $T$, i.e. for some $C_L \in \R$,
\begin{equation}
f_r(t,x_r[t]) = C_L\ \forall t \geq T
\end{equation}

In this paper, we consider the problem of sufficient graph theoretical conditions for consensus tracking among normally behaving agents to occur when these agents apply the W-MSR algorithm in the presence of misbehaving agents, which will be defined later in this section. As we will show in section \ref{sec:insuff}, the graph theoretic concepts of $r$- and $(r,s)$-robustness are insufficient to guarantee consensus tracking under these conditions.

Three types of agents are considered: normal agents, agents behaving as leaders, and misbehaving or adversarial agents. Normal agents update their state by applying the W-MSR algorithm, which was defined in \cite{LeBlanc_2013_Res} and is summarized below:

\begin{enumerate}
\item Each agent $i$ receives state values from its inclusive set of in-neighbors $\J_i$ and forms a sorted list.
\item Each agent $i$'s own state value is used as a reference to filter out outliers in the sorted list. If there are $F$ or less values greater than (or less than) its state value, each agent removes all states greater than (less than) its own from the list. If there are $F$ or more states greater than (or less than) its own state, it removes the highest (lowest) $F$ states.
\item Using the remaining values, denoted as $\mathcal{R}_i[t]$, each agent $i$ updates its state according to a function of the remaining states:
\begin{equation}
\label{eq:wmsr}
x_i[t+1] = \sum_{j \in \mathcal{R}_i[t]} w_{ij}[t]x_j[t],\ \forall i \in \N
\end{equation}
\end{enumerate}

The set of normal agents is denoted $\N$. It is assumed that the following conditions hold for the weights $w_{ij}[t]$ for all $i \in \N$ and for all $t \in \mathbb{Z}_{\geq 0}$:

\begin{itemize}
\item $w_{ij}[t] = 0$ when $j \notin \mathcal{J}_i$
\item $w_{ij}[t] \geq \alpha,\ 0 < \alpha < 1,\ \forall j \in \mathcal{J}_i$
\item $\sum_{j=1}^n w_{ij}[t] = 1$
\end{itemize}
Next, we define a agent behaving as a leader as follows:
\begin{define}
An agent behaves as a leader if it updates its state value according to 
\begin{equation}
\label{eq:leader}
x[t+1] = x_r[t+1] = f_r(t,x_r[t])\, \forall t \geq t_0
\end{equation}
\end{define}

We denote the set of such agents as $\Le \subset \V$. By a slight abuse of language, we will sometimes refer to agents behaving as leaders simply as ``leaders" in this paper. Such agents are not explicitly recognized by other agents in the network; their communicated state values are filtered and weighted in the same way as other agents' states in state update equations.

Finally, we define misbehaving agents as follows:
\begin{define}
An agent $j \in \V$ is misbehaving if both of the following conditions are satisfied:
\begin{enumerate}
\item There exists a time $t$ where agent $j$ does not update its state according to equation (\ref{eq:leader}) 
\item There exists a time $t$ where agent $j$ does not update its state according to \eqref{eq:wmsr} and/or sends different state values to different out-neighbors at any time $t$.
\end{enumerate}
\end{define}
Misbehaving agents therefore include agents that are Byzantine and malicious, as defined in \cite{LeBlanc_2013_Res,Zhang2012c}. Briefly, the key difference between the two is that malicious agents send the same state value to all out-neighbors while Byzantine agents may send different values to different out-neighbors. We denote the set of misbehaving agents as $\A$. We assume that $\N \cup \Le \cup \A = \V$. The threat model assumed in this paper is that $\A$ is an \emph{$F$-local set}, meaning that for any $i \in \N$, $|\J_i \cap \A| \leq F$ \cite{Zhang2012robustness}.

The following Lemmas will be needed in our analysis of conditions guaranteeing resilient consensus tracking:


\begin{lemma}
Consider the reference tracking consensus problem. A necessary condition for the normal agents reaching consensus to the group reference at any arbitrary value under the W-MSR protocol is the presence of at least $F+1$ agents behaving as leaders.
\end{lemma} 

\begin{proof}
Suppose $|\Le| \leq F$, and all other normal and misbehaving agent states are at a common value $a \in \R$. Suppose $x_r[t_1] = a_1, a_1 \neq a$ for some $a_1 \in \R$ and $t_1$. If the malicious agents keep their state values at $a$, then $\forall i \in \N,\ |\mathcal{R}_i[t] \cap \Le| = 0\ \forall t \geq t_1$ by definition of the W-MSR algorithm. Thus all normal agents will remain at $a$.
\end{proof}

\begin{lemma}
\label{lem:constL}
Under the W-MSR protocol, any agent $i$ connected to at least $F+1$ agents behaving as leaders with state values equal to a value $C_L$ will receive input from those leaders at all times. 
\end{lemma}

\begin{proof}
Suppose an agent $i$ is connected to all agents in a set of leaders $S_L$, $|S_L| \geq F+1$ whose state value is $C_L$. If $C_L$ is greater or less than $x_i[t]$, agent $i$ will filter out at most $F$ values from $S_L$. If $C_L = x_i[t]$, none of the states from $S_L$ will be filtered out.
\end{proof}

\subsection{Insufficiency of $r$- and $(r,s)$-Robustness for Consensus to Reference Values} \label{sec:insuff}

In \cite{LeBlanc_2013_Res} it was shown that under an $F$-total  malicious adversary model, a necessary and sufficient condition for resilient asymptotic consensus in a network applying the W-MSR protocol is $(F+1,F+1)$-robustness of the network. Similarly, under an $F$-local malicious or $F$-total Byzantine adversary model, a sufficient condition for resilient asymptotic consensus is ($2F+1$)-robustness of the network. In order to achieve resilient consensus tracking, one might be tempted to take a ($2F+1$)- or $(F+1,F+1)$-robust network and designate any $F+1$ arbitrary agents to behave as leaders. We demonstrate that this is not sufficient to guarantee that the normal agents will achieve consensus to the leader agents.

Consider an $(F+1,F+1)$-robust network partitioned into two disjoint sets $S_1$ and $S_2$, $S_1 \cup S_2 = \V$. Let $a_1,a_2 \in \R,\ a_1 \neq a_2$ be two arbitrary values. Define $\mathcal{X}_{S_m}^r = \{i \in S_m: | \V_i \setminus S_m| \geq r\}$ for $m \in \{1,2\}$. Suppose the set $S_1$ consists of $F+1$ agents with initial value of $a_1$, all of which have $F+1$ in-neighbors outside of $S_1$ ($|\Xr{F+1}{1}| = |S_1|$). Suppose the set $S_2$ consists of all other agents in the network, all with initial value $a_2$ and none of which have more than $F$ connections outside of $S_2$ ($|\Xr{F+1}{2}| = 0$). Let all agents in $S_1$ behave as leaders with $x_r[t] = a_1$. This graph meets the criteria for being $(F+1,F+1)$-robust, but the normal agents will not achieve consensus to the leaders since none have at least $F+1$ in-neighbors outside of $S_2$. 


Next, consider a ($2F+1$)-robust network with a set $S_1$ of $F+1$ agents with state values at $a_1$. Let $S_2$ contain all other agents in the graph with initial state values of $a_2$. Since the graph is ($2F+1$)-robust, there necessarily exists at least one agent in $S_1$ that has $(2F+1)$ connections in $S_2$. Suppose that $F$ or less agents in $S_2$ are connected to all agents in $S_1$. This graph satisfies the conditions for $(2F+1)$-robustness. But, if all agents in $S_1$ behave as leaders with $x_r = a_1$, and if all $F$ or less agents in $S_2$ connected to all the leaders in $S_1$ become malicious and hold their state values at $a_2$, the normal agents will never achieve consensus to the leaders. This demonstrates that simply designating $F+1$ arbitrary agents in a $(2F+1)$- or $(F+1,F+1)$-robust network to act as the group reference will not guarantee that all normal agents will converge to $x_r[t]$. 

\section{Consensus to a Reference Value in Robust Networks}

\subsection{Consensus Tracking Without Trusted Agents}

It is clear that different graph conditions are needed to guarantee consensus to the group reference under the W-MSR algorithm. We will demonstrate that the property of \emph{strong $r$-robustness} is sufficient for this to be achieved.

\begin{define}
	\cite{mitra2018} Let $r \in \Z_{\geq 0}$ and $S \subset \V$. A graph is strongly $r$-robust with respect to $S$ if for all nonempty subsets $C \subseteq \V \backslash S$, $C$ is $r$-reachable.
\end{define}




\begin{lemma}
\label{lem:mbarmono}
Consider a graph $G = (\V,\E)$ which satisfies an $F$-local adversary model with $|\Le \cap \A| \geq 0$. Suppose $G$ is strongly $(2F+1)$ robust with respect to $\Le$ and the normal agents apply the W-MSR algorithm. Define $\overline{m}[t] = \min_{i \in \N} (x_i[t],x_r[t])$ and $\overline{M}[t] = \max_{i \in \N} (x_i[t],x_r[t])$. Then on any time interval $[t_1,t_2)$ where $x_r[t]$ is constant, the following hold $\forall t \in [t_1,t_2)$:
\begin{itemize}
\item $x_i[t] \in [\mb[t_1],\Mb[t_1]]\ \forall i \in \V \backslash \A$
\item $\mb[t]$ and $\Mb[t]$ are nondecreasing and nonincreasing, respectively
\end{itemize}
\end{lemma}

\begin{proof}
Consider an agent $i \in \N$. At any time $t \in [t_1,t_2)$, agent $i$ will receive at most $F$ values from $\J_i \cap \A$, denoted $\{x_{a_1}[t],\ldots,x_{a_j}[t]\},\ j \leq F$, by definition of an $F$-local adversary model. For any such $j$, observe that $x_{a_j}[t] > \Mb[t]$ implies $x_{a_j}[t] > x_k[t],\ \forall k \in \J_i \backslash \A$. This implies that $x_{a_j}[t]$ will be one of the highest values in $i$'s sorted list of state values, and therefore will not be in $\mathcal{R}_i[t]$. A similar method can be used to show $x_{a_j} < \mb[t]$ implies $x_{a_j}[t] \notin \mathcal{R}_i[t]$. Therefore by definition of the W-MSR algorithm, $x_i[t]$ will update its state with a convex combination of values in $[\mb[t],\Mb[t]]$, implying $\mb[t+1] \geq \mb[t],\ \Mb[t+1] \leq \Mb[t]\ \forall t \in [t_1,t_2)$ (since $x_r[t]$ is constant on $[t_1,t_2)$). Therefore $x_i[t] \in [\mb[t_1],\Mb[t_1]]$ and $\mb[t]$ and $\Mb[t]$ are nondecreasing and nonincreasing, respectively, for $t \in [t_1,t_2)$.

\end{proof}

 Lemma \ref{lem:mbarmono} gives a safety condition in the sense that when $x_r[t]$ is constant $\forall t \in [t_1,t_2)$, $[\mb[t_1],\Mb[t_1]]$ is an invariant set on the same time interval.

%
%
%
%
%
%

\begin{theorem}
\label{constLeaderInConv}
Consider a graph of $n$ agents operating under the W-MSR algorithm that satisfies an $F$-local adversary model with $|\Le \cap \A| \geq 0$. Suppose further that there exists a time $t_C$ and a constant $C_L \in \R$ such that $x_r[t]=C_L\ \forall t \geq t_C$. Then $\lim_{t \rightarrow \infty} \|x_i[t] - x_r[t]\| = 0\ \forall i \in \N$ if the network is strongly $(2F+1)$-robust with respect to $\Le$.
\end{theorem}

\begin{proof}
We prove by contradiction. By Lemma \ref{lem:mbarmono}, we know that $\overline{m}[t]$ and $\overline{M}[t]$ are monotonic functions for all $t \geq t_C$ and therefore have limits $A_{\mb}$ and $A_{\Mb}$, respectively. The agents reach consensus to $x_r[t]$ if $A_\mb = A_\Mb = C_L$. 

We first consider the case when $\Mb[t_C] > C_L$ and $\mb[t_C] < C_L$. Let $\eps_{m0} = C_L - A_\mb$ and $\eps_{M0} = A_\Mb - C_L$. Let $\eps_0 = \min(\eps_{m0},\eps_{M0})$. This implies that $A_\mb + \eps_0 \leq C_L \leq A_\Mb - \eps_0$. Define the functions $X_m(t,\epsilon_i) = \{i \in \N: x_i[t] < A_\mb + \epsilon_i\}$ and $X_M(t,\epsilon_i) = \{i \in \N: x_i[t] > A_\Mb - \epsilon_i\}$. Next, define the set $S_X(t,{\eps_i}) = X_m(t,\eps_i) \cup X_M(t,\eps_i)$. Also define $\epsilon < \frac{\alpha^N}{1-\alpha^N}\epsilon_0$ where $N = |\N|$. Note that $\epsilon_0 > \epsilon > 0$. By the definition of convergence, there exists a time $t_\epsilon$ such that $\mb[t_\eps] > A_\mb - \eps$ and $\Mb[t_\eps] < A_\Mb + \eps$. Consider the set $S_{N_0} = \N$.  By definition of strong $(2F+1)$-robustness, there exists at least one agent $i \in S_{N_0}$ that is connected to at least $F+1$ leader agents. This is because at least one agent in $S_{N_0}$ must have $2F+1$ in-neighbors outside of $S_{N_0}$. Since an $F$-local adversary model implies $|\J_i \cap \A| \leq F$, $i$ must be connected to at least $F+1$ leaders. From Lemma \ref{lem:constL} we know that $i$ will use at least one leader's state to update its own state for all time. 

By Lemma \ref{lem:mbarmono}, agent $i$ will update its state with values on the interval $[\mb[t_\eps,\Mb[t_\eps]]$. To demonstrate the smallest and largest values that $i$'s state can possibly have at time $t_\eps + 1$, we assume that the minimum possible weight $\alpha$ is placed on the leader's state value, and the maximum possible weight is placed on $\mb[t_\eps]$ or $\Mb[t_\eps]$, the smallest and largest values possible for agent $i$'s in-neighbors. We have
\begin{align*}
x_i[t_\eps +1] &\geq (1-\alpha)\mb[t_\eps] + \alpha C_L\\
&\geq (1-\alpha)(A_\mb - \eps) + \alpha(A_\mb + \eps_0) \\
&\geq A_\mb + \alpha \eps_0 - \eps (1-\alpha) \\
x_i[t_\eps +1] &\leq (1-\alpha)\Mb[t_\eps] + \alpha C_L\\
&\leq (1-\alpha)(A_\Mb + \eps) + \alpha(A_\Mb - \eps_0) \\
&\leq A_\Mb - \alpha \eps_0 + \eps (1-\alpha)
\end{align*}
Define $\eps_1 = \alpha \eps_0 - \eps(1-\alpha)$. It can be shown that $\eps_0 > \eps_1 > \eps >0$. The above bounds also apply to any other normal agents connected to at least $F+1$ leader agents, and therefore each normal agent $i$ connected to at least $F+1$ leaders satisfies $A_\mb + \eps_1 \leq x_i[t_\eps+1] \leq A_\Mb - \eps_1$. Since these agents use their own state when performing state updates, and by monotonicity of $\mb[t]$ and $\Mb[t]$, these bounds also apply $\forall t \geq t_\eps +1$. Now, observe that $|S_X(t_\eps+1,\eps_1)| \leq |S_X(t_\eps,\eps)|$. Consider the new set $S_{N_1} = S_X(t_\eps+1,\eps_1)$. Observe that all normal agents connected to at least $F+1$ leaders are not in $S_{N_1}$. By the conditions of $(2F+1)$-robustness, there exists at least one agent $i_1$ that will not filter out at least one in-neighbor with state value within the bounds $[A_\mb + \eps_1, A_\Mb - \eps_1]$. Therefore the lowest and highest bounds for $x_{i_1}[t_\eps + 2]$ are
\begin{align*}
x_{i_1}[t_\eps+2] &\geq (1-\alpha)\mb[t_\eps]+\alpha(A_\mb+\eps_1) \\
&\geq A_\mb + \alpha \eps_1 - \eps (1-\alpha) \\
x_{i_1}[t_\eps +2] &\leq (1-\alpha)\Mb[t_\eps] + \alpha (A_\Mb - \eps_1)\\
&\leq A_\Mb - \alpha \eps_1 + \eps (1-\alpha)
\end{align*}
These bounds also apply for all normal agents not in $S_{N_0}$ since each normal agent uses its own state value when performing its state update. Defining the new value $\eps_2 = \alpha \eps_1 - \eps (1-\alpha)$, we then have $i_1 \notin S_X(t_\eps + 2,\eps_2)$. This implies  ${|S_X(t_\eps+2,\eps_2)|} < |S_X(t_\eps+1,\eps_1)|$. We now define $\eps_j$ recursively for $j \geq 1,\ j \in \Z$ as $\eps_j = \alpha \eps_{j-1} - (1-\alpha)\eps$. It can be shown that $\eps_j < \eps_{j-1}\ \forall j \geq 1$. If at each timestep $t_\eps + j$ we define $S_{N_j} = S_X(t_\eps+j,\eps_j)$, there will exist a agent $i_j \in S_{N_j}$ with at least $F+1$ in-neighbors in the interval $[A_\mb+\eps_j, A_\Mb-\eps_j]$. This implies that $i_j$ has an in-neighbor $k \in S_L: A_\mb + \eps_j \leq x_k[t_\eps + j] \leq A_\Mb + \eps_j$ that will not be filtered out when updating its state for time step $t_\eps + j+1$. Therefore agent $i_j$ will not be in $S_{N_{j+1}}$ for time step $t_\eps+j+1$ (as well as all normal agents which were not in $S_{N_j}$.

Continuing the above analysis for each time step $t_\eps+j$, we see that $|S_X(t_\eps+j,\eps_j)| < |S_X(t_\eps+j-1,\eps_{j-1})|$. Since $|S_X(t_\eps,\eps_0)| \leq |\N|$, there exists a time step $t_\eps + T,\ T \leq |\N|$ where $|S_X(t_\eps+T,\eps_T)| = 0$. This implies that $\mb[t_\eps+T] \geq A_\mb + \eps_T$ and $\Mb[t_\eps+T] \leq A_\Mb - \eps_T$. We demonstrate that $\eps_T > 0$, which contradicts $A_\mb$ and $A_\Mb$ being the lower and upper limits for $\mb[t]$ and $\Mb[t]$. Let $N = |\N|$. The value of $\eps_T$ satisfies
\begin{align*}
\eps_T &= \alpha \eps_{T-1} - (1-\alpha)\eps \\
&= \alpha^2 \eps_{T-2} - \alpha(1-\alpha)\eps - (1-\alpha)\eps \\
& \vdots \\
&= \alpha^T \eps_0 - (1-\alpha)(1+\alpha+\ldots+\alpha^{T-1})\eps\\
&= \alpha^T \eps_0 -(1-\alpha^T)\eps \\
&\geq \alpha^N \eps_0 - (1-\alpha^N)\eps
\end{align*}
By definition of $\eps$, $(1-\alpha^N)\eps < \alpha^N \eps_0$, which implies that $\eps_T > 0$. This contradicts our assumptions that $A_\mb < C_L$ and $A_\Mb > C_L$, proving that $A_\mb = A_\Mb = C_L$ and $\lim_{t \rightarrow \infty} \|x_i[t] - x_r[t]\| = 0\ \forall i \in \N$.

Finally, we consider the case when either $\Mb[t] = C_L$ or $\mb[t] = C_L$ for some $t \geq t_T$. Without loss of generality, consider the case where $\Mb[t] = C_L$. By Lemma \ref{lem:mbarmono} we know that all normal agents' state values are bounded by $[\mb[t_0],C_L]$ and that $\mb[t]$ is nondecreasing. Therefore m[t] has an upper limit $A_\mb$. Suppose $A_\mb < C_L$. Let $\epsilon_0 = C_L - A_\mb$. A similar strategy can be used to demonstrate that there exists a time where all agents have a value greater than $A_\mb$, contradicting that $A_\mb$ is an upper limit to $\mb[t]$ and proving that $\lim_{t \rightarrow \infty} \|x_i[t] - x_r[t]\| = 0\ \forall i \in \N$.



\end{proof}

%
%
%
%
%
%

\subsection{Consensus Tracking With Trusted Agents}

We point out that for a graph to be strongly $(2F+1)$-robust with respect to the set $\Le$, $|\Le| \geq 2F+1$. If we incorporate the notion of \emph{trusted agents} into the network, this lower bound on the number of leaders can be slightly relaxed. Trusted agents as defined in \cite{Abbas2014} are agents who have been made sufficiently secure to safely assume that they cannot be compromised and will not misbehave. We define \emph{trusted leader-follower robustness} as follows:

\begin{define}
\label{def:TNLF}
A network is \emph{trusted leader-follower robust} (TLF robust) with parameter $F \in \Z_{\geq 0}$ if for a set $S \subset \V$ and for all nonempty sets $C \subseteq \V \backslash S$, at least one of the following holds:
\begin{itemize}
\item There exists $i \in C$ with at least $F+1$ in-neighbors from $S$; i.e. $|\V_i \cap \mathcal{S}| \geq F+1$
\item $C$ is $2F+1$-reachable
\end{itemize}
\end{define}

Under TLF robustness, the minimum number of leaders required is only $F+1$. The following Lemma and Theorem extend the results obtained for $(2F+1)$-robust networks to TLF-robust networks:



\begin{lemma}
\label{lem:LFmbarmono}
Consider a graph $G = (\V,\E)$ which satisfies an $F$-local adversary model where $|\Le \cap \A| = 0$. Suppose $G$ is TLF robust with parameter $F$, and the normal agents apply the W-MSR algorithm. Define $\overline{m}[t] = \min_{i \in \N}(x_i[t],x_r[t])$ and $\overline{M}[t] = \max_{i \in \N} (x_i[t],x_r[t])$. Then on any time interval $[t_1,t_2)$ where $x_r[t]$ is constant, the following hold $\forall t \in [t_1,t_2)$:
\begin{itemize}
\item $x_i[t] \in [\mb[t_1],\Mb[t_1]]\ \forall i \in \V \backslash \A$
\item $\mb[t]$ and $\Mb[t]$ are nondecreasing and nonincreasing, respectively
\end{itemize}
\end{lemma}

\begin{proof}
The result can be shown by a method similar to the proof of Lemma \ref{lem:mbarmono}. 
\end{proof}

\begin{theorem}
Consider a network of $n$ agents with trusted leaders $(|\Le \cap \A| = 0)$ operating under the W-MSR algorithm that satisfies an $F$-local adversarial condition. Suppose further that there exists a time $t_C$ and a constant $C_L \in \R$ such that $x_r[t]=C_L\ \forall t \geq t_C$. Then $\lim_{t \rightarrow \infty} \|x_i[t] - x_r[t]\| = 0\ \forall i \in \N$ if the network is TLF robust with parameter $F$.
\end{theorem}

\begin{proof}
The result can be shown by using a method similar to Theorem \ref{constLeaderInConv}.

\end{proof}

\section{Application to Circulant graphs}

Circulant graphs are useful in the context of robust networks since their $r$-robustness is easily determined. Since $2k$-connected ring graphs are at least $\floor{\frac{k}{2}}$-robust (\cite{Zhang2015a,Salda2017}), an undirected circulant graph of the form $C_n\{\pm 1, \pm 2, \ldots, \pm k\}$ is $2k$ connected and therefore $k$-robust. In addition, our recent work in \cite{Usevitch2017} demonstrates that circulant digraphs of the form $C_n\{1,2,\ldots,k\}$ are at least $\ceil{\frac{k}{2}}$-robust.

We now show that these graphs can demonstrate strong $(2F+1)$-robustness and TLF robustness with parameter $F$ by selecting a proper subset of agents to behave as leaders. As per the definition of circulant graphs, we assume all agents are indexed $1,\ldots,n$. In our next proof we refer to sets of consecutive agents by index. An example is $P_L = \{2,3,4,5,\ldots,9\}$ in a network of $n=15$ agents. Since the index numbers are defined on the set of integers modulo $n$, the set $P_L = \{14, 15,1,2\}$ would also be a set of consecutive agents in a network of $n=15$ agents.

\begin{theorem}
\label{thm:kcirc}
A $k$-circulant digraph $\D = C_n\{1,2,\ldots,k\}$ is strongly $(2F+1)$-robust with respect to $\Le$ if $\D$ contains a set of consecutive agents by index $P_L$ such that $|P_L| \leq k$ and $|P_L \cap \Le| \geq 2F+1$. Moreover, $\D$ satisfies the conditions for TLF robustness with parameter $F$ if $|P_L| \leq k - F$ and $|P_L \cap \Le| \geq F+1$.
\end{theorem}

\begin{proof}
Since robustness analysis occurs prior to knowing which agents will become adversarial, all agents will be treated as either normal or leaders. Suppose $k \geq |P_L|$ and $|P_L \cap \Le| \geq 2F+1$. Let the first agent in $P_L$ be labeled as agent $(n-|P_L|+1)$ and the last agent in $P_L$ as agent $n$. We must show that all nonempty $C \subseteq \V \backslash \Le$ are $(2F+1)$-reachable. If agent $1 \in C$ then $C$ is $(2F+1)$-reachable since $ \{(n-|P_L|+1),\ldots,n \} \subseteq \V_{1}$ which implies $|\V_{1} \cap (\V \backslash C)| \geq 2F+1$. Next, suppose that agent $1 \notin C$ and $2 \in C$. Since $\{(n-|P_L|+2),\ldots,1 \}\subseteq \V_{2}$, this implies that $|\V_{2} \cap (\V \backslash C)| \geq |\V_{1} \cap (\V \backslash C)| \geq 2F+1$ and therefore $C$ is $2F+1$-reachable. This reasoning can be continued inductively by assuming $\{1,\ldots p-1\} \notin C$, $p \in C$, and observing that $|\V_{p} \cap (\V \backslash C)| \geq |\V_{p-1} \cap (\V \backslash C)|$. Note that if $p$ is ever the number of an agent in $\Le$, then we need not consider it ever being in $C$ and the analysis can be continued with the next normal agent. This reasoning can be continued through the entire network to show that all nonempty $C$ are $(2F+1)$-reachable.

To prove the result for TLF robustness with parameter $F$, assume that agents are labeled in the same way as before. Note that $\{(n-|P_L|+1),\ldots,n\} \subseteq \V_j,\ \forall j \in \{1,\ldots,F\}$. Therefore agents 1 through $F$ each have at least $F+1$ leaders as in-neighbors, i.e. $|\V_j \cap \Le| \geq F+1\ \forall j \in \{1,\ldots,F\}$. Therefore, if any of these agents are in $C$, the network satisfies the first condition of TLF robustness with parameter $F$. Now, assume that $\{1,\ldots,F\} \notin C$ and $F+1 \in C$. Note that $\{(n-|P_L|+1),\ldots,n,1,\ldots,F\} \subseteq \V_{F+1}$, and therefore $|\V_{F+1} \cap (\V \backslash C)| \geq 2F+1$, satisfying the second condition of TLF robustness. A similar inductive method as the one used to prove $(2F+1)$-robustness can then be used to show that the network is TLF robust with parameter $F$.
\end{proof}

As an example, suppose a $k$-circulant digraph has $n=10$ agents, $k = 7$, and $F = 2$ with agents $\{1,4,5\}$ being leaders. The set of consecutive agents $P_L = \{1,2,3,4,5\}$ satisfies $|P_L| \leq k - F$ and $|P_L \cap \Le| \geq F+1$. Therefore the digraph is TLF robust with parameter $F=2$.

\begin{theorem}
An undirected circulant graph of the form $\mathcal{G} = C_n\{\pm 1,\pm 2,\ldots,\pm k\}$ is strongly $(2F+1)$-robust with respect to $\Le$ if $\mathcal{G}$ contains a set of consecutive agents $P_L$ such that $|P_L| \leq k$ and $|P_L \cap \Le| \geq 2F+1$. Moreover, $\mathcal{G}$ satisfies the conditions for TLF robustness with parameter $F$ if $|P_L|\leq k-F$ and $|P_L \cap \Le| \geq F+1$.
\end{theorem}

\begin{proof}
The same method in Theorem \ref{thm:kcirc} can be applied to prove the result.
\end{proof}

Because $k$-circulant graphs demonstrate properties of $r$-robustness, strong $r$-robustness, and TLF robustness, they can be used both for situations requiring consensus tracking and for situations requiring only resilient asymptotic consensus under the W-MSR algorithm. We point out that the properties of strong $(2F+1)$-robustness and TLF robustness are a two-edged sword: if the $F$-local adversary assumption is violated in a strongly $(2F+1)$ or TLF robust network, a properly chosen set of misbehaving agents could potentially influence the network in a similar manner to a set of leaders.

\section{Simulation}

For our simulations, we consider agents in $k$-circulant digraphs. We assume the agents are trying to come to agreement on a state variable of interest such as altitude, the radius of a circular patrolling path (\cite{Salda2017}), minimum inter-agent separation distance, etc. The implementation of the W-MSR algorithm among all agents serves as a distributed consensus manager to allow agents to reach agreement on the variable of interest. Once agreement is reached, agents can apply local control laws based on this variable (\cite{Ren2007Inf}). We emphasize that in these simulations, agents have no knowledge as to whether their in-neighbors are normal, malicious, or behaving as leaders. In all simulations, for all agents in $(\V \backslash \Le)$ the agents' initial states are random values on the interval $[-25,25]$.

For comparison, Figure \ref{fig:Sim1} shows 20 agents running the normal W-MSR protocol in an 15-circulant digraph (which is 8-robust) with $F=3$. The 3 malicious agents are $\{1,6,15\}$, and are indicated by the dotted red lines. The normal agents come to consensus to a value within the convex hull of their initial states.

\begin{figure}
\includegraphics[width=\columnwidth]{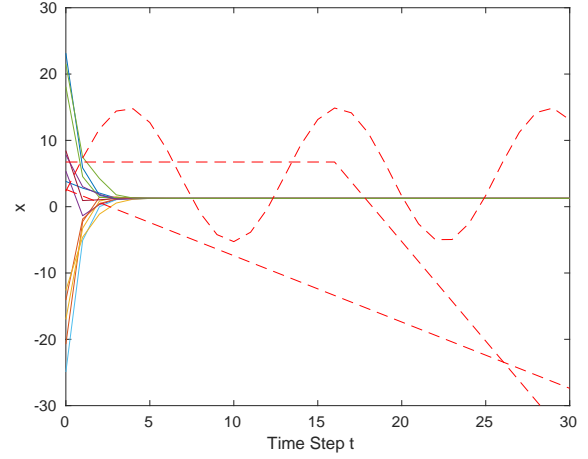}
\caption{A network of agents running the normal W-MSR algorithm with $n=20,\ k = 15$. The dotted red lines represent misbehaving agents, while the solid lines represent normally behaving agents.}
\label{fig:Sim1}
\end{figure}

The next simulation does not assume any trusted leaders ($|\Le \cap \A| \geq 0$). In this simulation (Figure \ref{fig:Sim3}) $n=30$, $k=15$, and $F=3$. The agents initially behaving as leaders are agents $\{22, 23, 24, 25, 26, 27,28\}$. If we let $P_L = \{22,\ldots,28\}$, we see that this set of consecutive agents satisfies $|P_L \cap \Le| \geq 2F+1$ and $|P_L| \leq k$, and is therefore strongly $7$-robust by Theorem \ref{thm:kcirc}. Here, leaders $\{22,26,28\}$ are attacked and begin misbehaving. However, the normal agents are still able to converge to $x_r[t]$, which is outside the interval $[-25,25]$.

\begin{figure}
\includegraphics[width=\columnwidth]{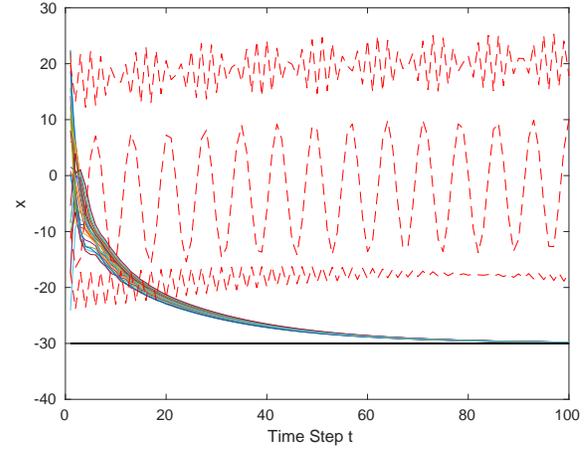}
\caption{Second simulation with $n=30$, $k=15$, and no trusted agents. The dotted red lines represent misbehaving agents, the solid black line represents $x_r[t]$, and the remaining solid lines represent normal agents.}
\label{fig:Sim3}
\end{figure}

The third simulation demonstrates a scenario where the reference value is initially changing but remains constant after a particular time. In this simulation $n=30$, $F=3$, $k=12$. The agents initially designated to behave as leaders are $\{22, 23, 24, 25, 26 ,27, 28\}$, which again indicates a strongly 7-robust graph by Theorem \ref{thm:kcirc}. Again, the leaders are not assumed to be trusted, and it is assumed agents $\{22,26,28\}$ misbehave. In this simulation, the value $x_r[t]$ switches between the values 30, -20, and 0, finally remaining at 0 for the remainder of the simulation. The normal agents asymptotically approach $x_r[t]$ on every interval where $x_r[t]$ is constant.

\begin{figure}
\includegraphics[width=\columnwidth]{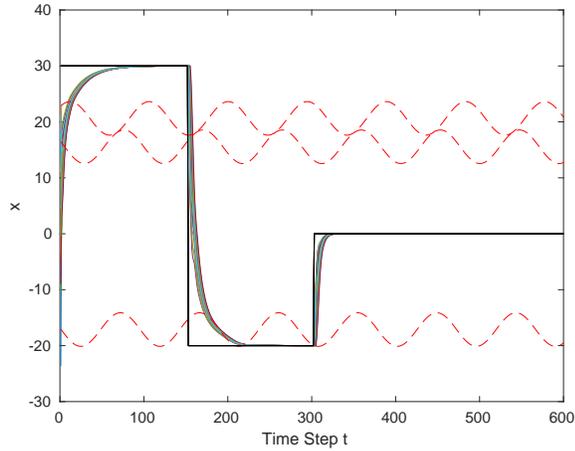}
\caption{Third simulation with $n=30$, $k=12$. Normal agents' state values converge to $x_r[t]$ on every interval where $x_r[t]$ is constant.}
\label{fig:Sim4}
\end{figure}

The fourth simulation demonstrates the same network and parameters as the third simulation, except the misbehaving agents demonstrate unbounded ramp function behavior. Despite this unbounded behavior, the normal agents still converge to $x_r[t]$.

\begin{figure}
\includegraphics[width=\columnwidth]{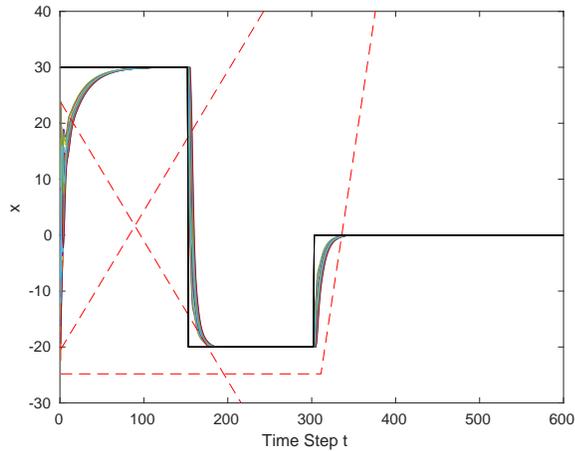}
\caption{Fourth simulation with similar parameters as the third simulation, except misbehaving agents are demonstrating unbounded ramp function behavior instead of sinusoidal behavior.}
\label{fig:Sim2}
\end{figure}

\section{Conclusion}

In this paper, we outline conditions for networks operating under the W-MSR algorithm to achieve consensus to arbitrary reference values. Our future work will involve extending these results to the case of time-varying graphs, and considering reference signals with different dynamics.

\bibliographystyle{IEEEtran}

\bibliography{Mendeley.bib}
\end{document}